\documentclass[11pt]{article}

\usepackage{array}
\usepackage{amsfonts,amsmath,amssymb,amsthm,mathrsfs}
\usepackage{adjustbox}
\usepackage{tikz}
\usepackage{enumitem}
\usepackage{ifthen}
\usepackage{fullpage}

\usepackage[ruled]{algorithm2e} 

\usepackage{algorithmic}
\SetAlFnt{\small}
\SetAlCapFnt{\small}
\SetAlCapNameFnt{\small}
\SetAlCapHSkip{0pt}
\IncMargin{-\parindent}

\usepackage{booktabs} 
\usepackage{xspace}
\usepackage[pagebackref=true,hidelinks]{hyperref}
\usepackage{bbm}
\usepackage{enumitem}
\usepackage[numbers]{natbib}
\usepackage{multirow}
\usepackage{graphicx}
\usepackage{subfigure}
\usepackage{wrapfig}
\usepackage{cleveref}

\newcommand{\D}{\mathcal{D}}
\newcommand{\C}{\mathcal{C}}
\renewcommand{\L}{\mathcal{L}}

\newcommand{\N}{\mathcal{N}}
\newcommand{\E}{\mathbb{E}}
\newcommand{\R}{\mathbb{R}}

\newcommand{\A}{\mathcal{A}}
\newcommand{\X}{\mathbf{X}}
\newcommand{\x}{\mathbf{x}}
\renewcommand{\P}{\mathbb{P}}
\newcommand{\PrC}{\textsc{PriceCplx}}
\newcommand{\SaC}{\textsc{SampleCplx}}
\newcommand{\sign}{\textsf{sign}}
\newcommand{\abs}[1]{\vert #1 \vert}
\newcommand{\reg}{\mathcal{C}_{\textsc{REG}}}
\newcommand{\alldist}{\mathcal{C}_{\textsc{ALL}}}
\newcommand{\mhr}{\mathcal{C}_{\textsc{MHR}}}
\newcommand{\eps}{\epsilon}
\newcommand{\rev}{\textsc{Rev}}
\newcommand{\one}{\mathbf{1}}

\newcommand{\dist}

\newcommand{\wrev}{\widetilde{\rev}}
\newcommand{\dkl}{D_{\textrm{KL}}}

\newtheorem{theorem}{Theorem}[section]
\newtheorem{heuristic algorithm}{Heuristic Algorithm}
\newtheorem{lemma}[theorem]{Lemma}

\newtheorem{corollary}[theorem]{Corollary}

\newtheorem{innercustomgeneric}{\customgenericname}

\providecommand{\customgenericname}{}
\newcommand{\newcustomtheorem}[2]{%
  \newenvironment{#1}[1]
  {%
   \renewcommand\customgenericname{#2}%
   \renewcommand\theinnercustomgeneric{##1}%
   \innercustomgeneric
  }
  {\endinnercustomgeneric}
}

\newcustomtheorem{customtheorem}{Theorem}
\newcustomtheorem{customlemma}{Lemma}
\newcustomtheorem{customcorollary}{Corollary}

\begin{document}

\title{Pricing Query Complexity of Revenue Maximization}
\author{
Renato Paes Leme \\ Google Research \\ {\tt renatoppl@google.com} \and 
Balasubramanian Sivan \\ Google Research \\ {\tt balusivan@google.com} \and
Yifeng Teng \\ Google Research \\ {\tt yifengt@google.com} \and 
Pratik Worah \\ Google Research  \\ {\tt pworah@google.com}
}
\date{}

\maketitle
\thispagestyle{empty}

\begin{abstract}
The common way to optimize auction and pricing systems is to set aside a small fraction of the traffic to run experiments. This leads to the question: how can we learn the most with the smallest amount of data? For truthful auctions, this is the \emph{sample complexity} problem. For posted price auctions, we no longer have access to samples. Instead, the algorithm is allowed to choose a price $p_t$; then for a fresh sample $v_t \sim \mathcal{D}$ we learn the sign $s_t = \sign(p_t - v_t) \in \{-1,+1\}$. How many pricing queries are needed to estimate a given parameter of the underlying distribution? 
 

We give tight upper and lower bounds on the  number of pricing queries required to find an approximately optimal reserve price for general, regular and MHR distributions. Interestingly, for regular distributions, the pricing query and sample complexities match. But for general and MHR distributions, we show a strict separation between them.

All known results on sample complexity for revenue optimization follow from a variant of using the optimal reserve price of the empirical distribution. In the pricing query complexity setting, we show that learning the entire distribution within an error of $\epsilon$ in Levy distance requires strictly more pricing queries than to estimate the reserve. Instead, our algorithm uses a new property we identify called \emph{relative flatness} to quickly zoom into the right region of the distribution to get the optimal pricing query complexity.

\end{abstract}

\newpage
\setcounter{page}{1}

\section{Introduction}
An important question in the intersection of economics and statistics is to estimate properties of a buyer's willingness-to-pay (a.k.a. value) from samples.  Samples typically come from buyer's bids in a truthful auction. In various economic setups though, collecting samples of true value is quite hard or impossible. For example, in many settings auction is not an option; often, posting prices is the only option. And the seller therefore has access only to a buyer's decision to buy or not at a given price. This is also true for auctions with very few buyers (not too uncommon in digital auctions as targeting becomes increasingly narrow and focused): if a buyer is the sole competitor, they may decide to bid just above the reserve price that is sent to them, or to not bid at all when their value is below the reserve price, in order to conceal information about their valuation. Likewise in non-truthful auctions, we don't have access to true value samples; we just have access to past bids. One thing we do have in all these scenarios is whether or not the true value exceeded a (reserve) price. Our goal in this paper to compare the power of \emph{value samples} (namely, samples from the buyer's value distribution) to that of \emph{pricing queries} (namely, a price $p$ and a response to whether or not the buyer's value is at least $p$).\\ 

What can we do with a single value sample? Dhangwatnotai et al.~\cite{DRY15} say that we can do a lot: we can price the buyer using that sample to extract \emph{half of the expected optimal revenue} achievable with full knowledge of the regular value distribution. Similarly a single value sample acts as an unbiased estimator of the mean of the distribution. Now consider a single pricing query. It is not possible to extract \emph{any non-zero fraction of the optimal revenue} from a single pricing query. Likewise there is nothing meaningful to estimate about the mean using a single pricing query. Both of these hold regardless of our choice of the price $p$. The contrast between the power of a single value sample and that of a single pricing query cannot be more stark. The question we investigate in this paper is what happens when we have a large number of pricing queries. How do the number of value samples required to get a $1-\epsilon$ approximation of the optimal revenue, or to estimate the mean within $\epsilon$ compare with the number of pricing queries required?\\

Let $\C$ be a class of distributions over $\R$ and $\theta:\C \rightarrow \R$ a parameter (such as the mean, the variance or the monopoly price). Let $\L:\R^2 \rightarrow \R$ be a loss function estimating the error in the parameter estimator\footnote{The loss can  implicitly depend on the distribution $\D \in \C$.}. We will consider the following parameters and respective losses:
\begin{itemize}
    \item \emph{Median:} $\theta(\D) = \text{median}(\D)$ with loss $\L(\theta; \hat\theta) = \abs{F_\D(\theta) - F_\D(\hat \theta)}$, where $F_\D$ is the c.d.f.
    \item \emph{Mean:} $\theta(\D) = \E_{v \sim \D}[v]$  with loss $\L(\theta; \hat\theta) = \abs{\theta - \hat \theta}$.
    \item \emph{Monopoly price:} $\theta(\D) = \text{argmax}_p \rev_\D(p) := \E_{v \sim \D} [p \cdot \one\{v \geq p\}]$ with the loss that compares the revenue at the optimal and estimated point: $\L(\theta; \hat \theta) = \abs{\rev_\D(\theta) - \rev_\D(\hat \theta)}$.
\end{itemize}
In our last example, the parameter is a function rather than a real number:

\begin{itemize}
    \item \emph{CDF:} $\theta(\D) = F_\D(\cdot)$  with loss $\L(F_\D; \hat F) = \textsf{LevyDistance}(F_\D; \hat F)$ which corresponds to the infimum over $\epsilon>0$ such that $F_\D(v-\epsilon)-\epsilon \leq \hat F(v) \leq F_\D(v+\epsilon)+\epsilon$ for all $v \in \R$.
\end{itemize}

A pricing estimator with $m$ queries consists of an algorithm that posts $m$ prices\footnote{We allow both values $v_t$ and prices $p_t$ to be negative. A negative value means that the buyer has a disutility for the good, while a negative price means that the buyer is compensated for acquiring the good. Negative values are allowed for convenience to obtain cleaner algorithms. The same results can be obtained for non-negative values by dealing with the corner cases around zero.} $p_1, p_2, \hdots, p_m \in \R$. For each price $p_t$ posted the algorithm observes if a buyer with a freshly drawn value buys or not at that price. In other words, we observe $s_t = \sign(v_t - p_t) \in \{-1,+1\}$ for some $v_t \sim \D$. The prices can be computed adaptively. Finally, the algorithm outputs an estimate $\hat \theta (s_1, \hdots, s_m)$. An $(\epsilon,\delta)$-pricing-query-estimator is an algorithm such that:
$$\P_{v_1, \hdots, v_m \sim \D} \left( \L(\theta(\D); \hat \theta (s_1, \hdots, s_m)) \leq \epsilon \right) \geq 1-\delta, \quad \forall \D \in \C$$

Given a class $\C$ and parameter $\theta$ and loss $\L$ we define the \emph{pricing query complexity} $\PrC_{\C,\theta}(\epsilon)$ to be the minimum $m$ such that there exists a $(\epsilon,1/2)$-pricing-query-estimator for $\C$, $\L$ and $\delta$.

It is useful to compare with traditional estimators and sample complexity. A traditional $(\epsilon,\delta)$-estimator is an algorithm that processes the samples $v_1, \hdots, v_m$ directly and produces an estimator $ \hat \theta (v_1, \hdots, v_m)$ such that:
$$\P_{v_1, \hdots, v_m \sim \D} \left( \L(\theta(\D); \hat \theta (v_1, \hdots, v_m)) \leq \epsilon \right) \geq 1-\delta, \quad \forall \D \in \C$$

Similarly, we can define the sample complexity as $\SaC_{\C,\theta}(\epsilon)$ as the minimum $m$ such that there exists a $(\epsilon,1/2)$-traditional-estimator for $\C$ and $\delta$. Since one can simulate a pricing estimator from a traditional estimator we have:
\begin{equation}
    \PrC_{\C,\theta}(\epsilon) \geq \SaC_{\C,\theta}(\epsilon)
\end{equation}
The main question we ask is how big is this gap? I.e. how much harder is it to estimate a parameter from pricing queries as compared to samples.

\paragraph{Practical Motivation and Relation to Bandits} In practice auctions are optimized by using a small fraction of traffic (say 1\%) to run experiments and the remaining 99\% runs the deployed production auction. After a period (say a few days) we use what we learned in the experimental slice to deploy in the main slice. Typically one does not care about the revenue in the experimental slice, but we do care about making that slice as small as possible. This translates directly to the goal of pricing query complexity: how to learn the most with the fewest possible queries?

The alternative to run a small experimental slice to collect data in order to optimize the main mechanism is to apply a bandit algorithm on the full traffic. This may be appealing in theory, but it is not viable in practice where upper bounds on experimental traffic is a reality. For one, bandit algorithms would lead to a more unstable system for the bidders. Furthermore, only learning in a small (typically random) experimental slice minimizes incentives for bidders to strategize against the learning algorithm, since any given query will only be used to learn a price for the next day with a very small probability.

This practical difference also explains how our objectives and guarantees differ from those of bandit learning for revenue optimization. While the bandit approach focuses on optimizing average revenue while learning, the sample/pricing complexity approach provides a guarantee on the loss of the final estimator with high probability, but doesn't care about the actual revenue obtained during the learning phase. This is very much in line with the literature on sample complexity.

\paragraph{Our Results} We give matching upper and lower bounds (up to $\text{polylog}(1/\epsilon)$ factors) for the pricing complexity of various parameters and classes of distributions. Our results are summarized in  Table~\ref{tab:results}. For each row of the table, we assume for the pricing complexity bounds that the parameter being estimated is in the range $[0,1]$. Since pricing queries only receive binary feedback, we need an initial range to be able to return any meaningful guarantee.

It is interesting to compare our bounds with the sample complexity for the same regime. Interestingly in many important cases, such as estimating the mean of a normal distribution or the monopoly price of a regular distribution, pricing queries don't pose any significant handicap when one looks at the asymptotic regime.  This is however not always the case: when computing the monopoly price of an MHR distribution, we establish a clear demarcation: we show a tight bound of $\widetilde{O}(1/\epsilon^2)$ for pricing complexity which is asymptotically higher than the known bound of $\widetilde{O}(1/\epsilon^{3/2})$ for sample complexity by Huang et al \cite{HMR18}. We also show a gap for computing the monopoly price for general distributions, while $\Omega(1/\epsilon^3)$ is required for pricing queries, $\widetilde{O}(1/\epsilon^2)$ is sufficient for value samples.


\begin{table}[h]
\centering
\begin{tabular}{ |c|c|c|c| }
\hline
Parameter & Distribution & Pricing Complexity & Sample Complexity \\
\hline
Median & General & $\tilde{O}(1/\epsilon^2)$, $ \Omega(1/\epsilon^2)$ & $O(1/\epsilon^2)$,$ \Omega(1/\epsilon^2)$ \\
Mean & Normal & ${O}(\sigma^2/\epsilon^2)$, $\Omega(\sigma^2/\epsilon^2)$ &  $O(\sigma^2/\epsilon^2)$, $\Omega(\sigma^2/\epsilon^2)$  \\
CDF 
& General & $\tilde{O}(1/\epsilon^3)$, $\Omega(1/\epsilon^3)$ &  $\tilde{O}(1/\epsilon^2)$, $\Omega(1/\epsilon^2)$  \\
CDF 
& Regular & $\tilde{O}(1/\epsilon^3)$, $\Omega(1/\epsilon^{2.5})$ &   $\tilde{O}(1/\epsilon^2)$  \\
Monopoly Price & MHR & $\tilde{O}(1/\epsilon^2), \Omega(1/\epsilon^2)$ & 
$\tilde{O}(1/\epsilon^{1.5}), \Omega(1/\epsilon^{1.5})$ \\
Monopoly Price & Regular & $\tilde{O}(1/\epsilon^2)$ , $\Omega(1/\epsilon^{2})$ & 
$\tilde{O}(1/\epsilon^{2})$ \\
Monopoly Price & General & $\tilde{O}(1/\epsilon^3)$, $\Omega(1/\epsilon^3)$ & 
$\tilde{O}(1/\epsilon^2)$, $\Omega(1/\epsilon^2)$
\\
\hline
\end{tabular}
\caption{Our upper and lower bounds for various estimators and comparison with sample complexity. In each case, the pricing complexity bounds assume that the parameter being estimated is in $[0,1]$. $\sigma$ is the standard deviation of the normal distribution.}
\label{tab:results}
\end{table}

\paragraph{Techniques: Mean and Median} As a warm up, we first study how to estimate the median of a distribution using pricing queries. The algorithm is a hybrid of binary search and the UCB algorithm. The algorithm both keeps a range $[\ell, r]$ that contains the median (as in binary search) and builds a confidence interval (like the UCB algorithm) to estimate the quantile of $p = (\ell + r)/2$. Whenever the confidence bound can safely separate the quantile of $p$ from $1/2$, a binary search step is performed. The algorithm requires very few queries to perform a binary search step that is far from the median, but  starts requiring more and more as we approach the median.

The algorithm is agnostic to $\epsilon$ and works in the streaming model: it keeps an estimate of the median through the execution that gets better as it receives more pricing query opportunities. Surprisingly, its performance matches what one could obtain from value samples.

We then use the algorithm to estimate the mean of a normal distribution. Remarkably, we obtain the same guarantee of $O(\sigma^2/\epsilon^2)$ (up to constant factors) as the optimal algorithm that has access to sample values. We obtain a guarantee of $\tilde{O}(\sigma^2/\epsilon^2)$ without needing to know $\sigma$ at all. To obtain a bound of $O(\sigma^2/\eps^2)$ the algorithm needs to know $\sigma$; or if the algorithm knows only an upper bound $\bar{\sigma}$ on $\sigma$, we can obtain a guarantee of $O(\bar{\sigma}^2/\eps^2)$.

\paragraph{Techniques: Monopoly Price} For the monopoly price our techniques are significantly more involved. At the heart of the analysis is a new property of regular distributions called \emph{relative flatness}, which is of potential independent interest. It says that if the revenue curve has roughly the same value at $4$ different equidistant points, it can't hide a point with very high revenue point in between those four points. This is not just a consequence of a single-peaked curve, as such curves can easily hide a high revenue point between four points of equal revenue. The concavity of the revenue curve in the quantile space is well known for regular distributions. The novelty is in finding the right property in the \emph{value space} to be used.

Using this property we design an algorithm that keeps a list of subintervals of $[0,1]$ to explore. Unlike binary search we are not able to keep a single interval. Once the algorithm explores an interval it may discard it, but it may also break it into multiple subintervals and add them to the list. The relative flatness is put to use in efficiently discarding intervals that don't contain the monopoly price. We design a potential function to control the number of sub-intervals to guarantee we explore a small number of them. The argument is delicate: even the next subinterval to pick needs care; picking an arbitrary subinterval to process does not yield the desired bound. This algorithm obtains a bound of $\tilde{O}(1/\epsilon^2)$ for regular distributions. 

We then turn our attention to MHR distributions. For value samples, it is known that learning the monopoly price for MHR distribution is easier than learning the monopoly price of regular distributions. Surprisingly, this is no longer the case for pricing complexity. To prove this fact, we design two distributions such that no price is an $\epsilon$-approximation simultaneously for both distributions even though the distributions have very close c.d.f. and quantiles. We then use an information-theoretic argument to show that the sequence of bits obtained by the pricing queries for both distributions will be indistinguishable with fewer than $\Omega(1/\epsilon^2)$ queries. We do so by bounding the KL-divergence between the distribution of bits observed by running any algorithm on both distributions.

\paragraph{Why Not Learn the CDF of the Whole Distribution?} Known results on sample complexity for revenue maximization typically use a variant of optimal reserve price of the empirical distribution. In the case of pricing queries, first, there is no clear notion of what the empirical distribution is. Second, even if one is tempted to learn the CDF of the entire distribution within an $\epsilon$ distance in Levy metric, we show that this is provably inferior (see Table~\ref{tab:results} for CDF Parameter) to the tight bounds we get. Our technique of quickly zooming into the right region of the distribution using the notion of relative flatness is crucial for getting tight bounds.

\paragraph{Related Work} 
The work that is closest to ours is the sample complexity of single buyer, single item revenue maximization by Huang et al.~\cite{HMR18}. Surprisingly Huang et al.~\cite{HMR18} show that the number of samples required to estimate the monopoly price to obtain a $1-\eps$ approximation to revenue is even smaller than the number of samples required to estimate the optimal revenue. Our detailed comparison of pricing query complexity with the sample complexity of Huang et al. has already been presented, and, the pricing complexity often matches the already surprisingly small sample complexity.

The literature on sample complexity has seen a lot of activity in the past decade. Particularly, in the single-parameter setting, starting with Elkind's work~\cite{Elkind07} on finite support auctions, there has been a lot of progress. Cole and Roughgarden~\cite{CR14} study the sample complexity of revenue maximization in single item settings with independent non-identical distributions, where a ``sample'' is an $n$-tuple when there are $n$ bidders in the auction. They give sample complexity as a polynomial in $n$ and $1/\eps$ when seeking to obtain a $1-\eps$ approximation to the optimal revenue. With this definition, in many settings the number of samples is not even a function of $n$. When one seeks a $1/4$ approximation, based on a generalization of Bulow and Klemperer's result~\cite{BK96} by Hartline and Roughgarden~\cite{HR09}, just a single sample would do. Fu et al.~\cite{fu2015randomization} shows that the revenue guarantees can be improved with randomization. When the distributions of the $n$ bidders are i.i.d., Dhangwatnotai et al.~\cite{DRY15} showed that poly($\eps$) samples are enough to get a $1-\eps$ approximation, and so is it in the case of unlimited supply (digital goods) setting as that can be boiled down to a single agent setting. Devanur et al.~\cite{DHP16} consider the single parameter setting where the buyer distributions are not in discrete categories, but in a continuum, represented as a signal in the hands of the auctioneer. The line of literature for single-parameter settings culminates with the work of Guo et al.~\cite{guo2019settling} which provides tight sample complexity for buyers with regular, MHR, or bounded values.  \cite{hu2021targeting} further studies a targeting sample complexity model where for each buyer distribution, it is allowed to specify a fixed size of quantile range to sample from.
Sample complexity results in multi-parameter settings are comparatively fewer. Morgenstern and Roughgarden~\cite{MR16} and Vasilis Syrgkanis~\cite{syrgkanis2017sample} study the sample complexity of learning simple auctions (item pricing, bundle pricing etc.) in multi-parameter settings via pseudo-dimension. Gonczarowski and Weinberg \cite{gonczarowski2018sample} obtains a polynomial sample complexity for up to $\eps$-optimal multi-item auctions for additive buyers with independent item values. Brustle et al. \cite{brustle2020multi} further extends the result to buyers with some specific types of correlated item values. Sample complexity of revenue-optimal item pricings was explored earlier by Balcan et al.~\cite{BBHM08}. Sample complexity of welfare-optimal item pricings was explored by Feldman et al.~\cite{FGL15} and Hsu et al.~\cite{HMRRV16}. 

Finally, we also note that pricing query complexity is different from sample complexity in PAC learning using threshold concepts i.e, the literature following the work of~\cite{valiant}, despite the resemblance in definitions. In PAC learning, the same concept may be used as a test on different samples. However, this is not the case in pricing complexity, here we never explicitly observe any sample. This means that the techniques and bounds for PAC learning are going to be closer to usual sample complexity results than to pricing complexity, as is borne out from the lower bounds.

\section{Warmup: Estimating the Median and the Mean}\label{sec:gaussian}
Before going to the technically more involved Section~\ref{sec:monopoly}, we begin with the estimation of the median and mean of distributions as a warmup. 
While the results in Section~\ref{sec:monopoly} are our main results, there are some nuances in this Section~\ref{sec:gaussian} as well, including the removal of log factor in Theorem~\ref{lm:norm}.
\subsection{Median of a General Distribution}
\label{sec:median-general}

We start with one of the most basic problems in statistics: estimating the median of a distribution. Let $C$ be the class of all distributions with finite variance and median in $[0,1]$. When we are estimating the median of a distribution $\D$ with cdf $F$, a natural loss function is $\L(\theta,\hat\theta)=|F(\theta)-F(\hat\theta)|$. Then we can bound the pricing query complexity as follows:

\begin{theorem}\label{thm:medianub}
Let $\C$ be a class of distributions over $\R$ with median in $[0,1]$. Define: $$a_F:=\sup_{F\in\C}\log{\textstyle|F^{-1}(\frac{1}{2}+\frac{\eps}{3})-F^{-1}(\frac{1}{2}-\frac{\eps}{3})|^{-1}}.$$ We have
\[\PrC_{\C,\theta}(\epsilon) \leq \tilde{O}(1/\epsilon^2)a_F\log a_F.\]
\end{theorem}

Observe that the bound depends on a parameter $a_F$ which measures how concentrated the function is around the median. For example, if there is an $O(\epsilon)$ probability mass exactly
at the median then $a_F = \infty$. In such case, in order to obtain $\L(\theta,\hat\theta)\leq O(\epsilon)$ we would need to find the the median \emph{exactly} which is impossible only using pricing queries. Thus the dependence on $a_F$.\\

The bound in Theorem \ref{thm:medianub} is achieved by a combination of binary search and the UCB algorithm, described in Algorithm \ref{alg:binary_ucb}. The algorithm starts with a range $[0,1]$ of the potential median, and repeatedly prices at the middle point $p$ of the range. Since the algorithm obtains only binary feedback, it needs to keep sending pricing queries at this point until it has enough information to decide with high probability whether $Q(p) = 1- F(p) < 1/2$ or $Q(p) > 1/2$. To do that we keep an estimate of the quantile $Q(p)$ based on the binary feedback together with a confidence interval for this estimate. Once $1/2$ leaves this confidence interval we know with high probability on which side of $p$ the median is, and we can do a binary search step. After the binary search step, we restart the UCB step.  

\begin{algorithm}[htb]
\caption{Binary search algorithm that returns $p^*$ with $|Q(p^*)-\frac{1}{2}|\leq\eps$ with prob $1-\delta$ }
\label{alg:binary_ucb}
{Initialize $\ell =0$, $r=1$, $p=1/2$; $n=0$, $k=0$, $p^* =1/2$, $\epsilon^* = 1$\;
\For {t=1,2,3... (each new arriving query)} {
    Increment number of samples $n = n+1$ \;
    Price at $p$, if sold, increment $k= k+1$ \;
    If $\epsilon_p = 18\sqrt{\log(t/\delta)/ n} < \epsilon^*$ update $p^* = p, \epsilon^* = \epsilon_p$ \;
    \If {$\abs{1/2 - k/n} > 12\sqrt{\log(t/\delta)/ n}$ (i.e. $1/2$ is outside the confidence bound)} {
        If $k/n < 1/2$, update $r = p$. Otherwise, $\ell=p$\;
        Update $k=n = 0$ (reset the counters) and set $p =(\ell+r)/2$\;
    }
    \textrm{Return} $p^*$ if $\eps^*\leq\eps$\;
}
}
\end{algorithm}

One important feature of our algorithm is that it uses no information about $\epsilon$ other than as a stopping condition. In fact, one can think of the algorithm in the \emph{streaming model}: it processes a stream of `pricing opportunities' and updates a constant-size set of parameters in each step. At each point in time, the algorithm keeps an estimate $p^*$ together with a confidence bound $\epsilon^*$, that keeps getting better as the algorithm receives more data. The analysis of the pricing complexity is deferred to Section~\ref{sec:proof-median-general}.

\subsection{Mean of a Normal Distribution}
\label{sec:mean-normal}

We now apply this idea to estimate the mean of a normal distribution using pricing queries. The idea can be extended to other families of parametric distributions, but we focus here on normals to allow a crisp comparison with sample complexity. Our main result is the following price complexity upper bound, which exactly matches the well known sample complexity lower bound of $\Omega(\sigma^2/\epsilon^2)$:

\begin{theorem}\label{lm:norm}
Let $\C_{\sigma} = \left\{ \N(\mu, \sigma) \text{ s.t. } \mu \in [0,1]\right\}$ be the class of normal distributions with variance $\sigma^2$ with the loss $\L(\theta, \hat \theta) = \abs{\theta - \hat \theta}$. Then the pricing complexity of computing the mean is:
$$\PrC_{\C_{\sigma},\mu}(\epsilon)= O(\sigma^2/\epsilon^2).$$
\end{theorem}

\begin{proof}
A direct application of Theorem \ref{thm:medianub}  leads to an algorithm with $\tilde{O}(\sigma^2/\epsilon^2)$ pricing query complexity. To estimate a point $\hat \theta$ such that $\abs{\hat \theta - \mu} \leq \epsilon$ for a distribution with mean $\mu$, it is enough to find a point $\hat\theta$ such that $\abs{F(\hat\theta) - F(\mu)} \leq \epsilon/\sigma$. For a normal distribution the parameter $a_F$ in Theorem \ref{thm:medianub} is $a_F = O(\log 1/(\epsilon \cdot \sigma))$. If $\sigma \geq \epsilon$, the statement of the theorem directly yields a  $\tilde{O}(\sigma^2/\epsilon^2)$. If $\sigma < \epsilon$ on the other hand, all the points $p \notin [\mu-\epsilon, \mu+\epsilon]$ are such that $\abs{Q(p) - 1/2} = \Omega(1)$ and hence the UCB step in Algorithm \ref{alg:binary_ucb} takes a constant number of queries at each of those points. Hence in $O(\log (1/\epsilon))$ we query a point that is $\epsilon$-close to the mean. This algorithm is robust in the sense that it doesn't require us to know the variance as it is only used in the analysis.\\

Assume for now that we know exactly what the variance $\sigma^2$ is.  If we allow the algorithm to use the exact value of the variance $\sigma^2$, we can improve the guarantee from  $\tilde{O}(\sigma^2/\epsilon^2)$ to $O(\sigma^2/\epsilon^2)$. We can apply the following procedure:
\begin{itemize}
    \item Use Algorithm \ref{alg:binary_ucb}  with $\epsilon = 1/4$ to find a point $p$ with quantile $Q_{\mu, \sigma}(p) \in [1/4, 3/4]$.
    \item Using $O(\sigma^2/\epsilon^2)$ pricing queries, find an estimate $\hat q$ such that $\abs{\hat q - Q_{\mu,\sigma}(p)} \leq \epsilon/\sigma$.
    \item Solve the equation $\hat q = Q_{\hat \mu, \sigma}(p)$ to find an estimate $\hat \mu$ for the mean.
\end{itemize}

Note that since $Q_{\hat \mu, \sigma}(p) = Q_{0,1}(\frac{p-\mu}{\sigma})$ where $Q_{0,1}$ is the quantile function for the standard Gaussian, the solution to the equation in last step is: $$\hat{\mu} = p - \sigma \cdot Q_{0,1}^{-1}(\hat q)$$
Since the derivatives of $Q_{0,1}^{-1}$ are bounded in $[1/4, 3/4]$ an error of $\epsilon/\sigma$ in the estimation $\hat{q}$ leads to an error of $O(\epsilon/\sigma)$ in $Q_{0,1}^{-1}(\hat q)$. Hence $\abs{\mu - \hat \mu} \leq O(\epsilon)$.\\
\end{proof}

We can relax the assumption that we know $\sigma$ exactly and assume we only know an upper bound $\bar \sigma$. The proof is deferred to Section~\ref{sec:proof-mean-normal}.

\begin{corollary}\label{cor:median-ub}
Let $\bar\C_{\bar\sigma} = \left\{ \N(\mu, \sigma) \text{ s.t. } \mu \in [0,1],\ \sigma\leq\bar\sigma\right\}$ be the class of normal distributions with variance at most $\bar\sigma^2$ with the loss $\L(\theta, \hat \theta) = \abs{\theta - \hat \theta}$. Then the pricing complexity of computing the mean is:
$$\PrC_{\bar\C_{\bar\sigma},\mu}(\epsilon)= O(\bar\sigma^2/\epsilon^2).$$
\end{corollary}

\section{Estimating the Monopoly Price}
\label{sec:monopoly}
In this section, we focus on estimating the monopoly price (i.e., the revenue optimal price or the Myerson price) and measure our loss via the natural metric of the revenue gap between the true revenue and that obtained by our estimated price: $\L(\theta; \hat \theta) = \abs{\rev_\D(\theta) - \rev_\D(\hat \theta)}, \text{ where } \rev_\D(p) := \E_{v \sim \D} [p \cdot \one\{v \geq p\}]$, where $\theta$ and $\hat \theta$ are the true and estimated monopoly prices. 

An important tool in this section is the following lemma, which shows that to estimate the quantile $Q(p)=1-F(p)$ of price $p$ with additive error $\eps$, it suffices to use $\tilde{O}(\frac{1}{\epsilon^2})$ pricing queries on $p$. The proof is relatively standard, and is deferred to Section~\ref{sec:proof-monopoly}.

\begin{lemma}\label{lem:repeatprice}
For any value distribution $F$ supported in $[0,1]$ and value $p$, there is an algorithm $\textsf{EQ}_F(p,\epsilon,\delta)$ that makes $m=\tilde{O}(\frac{1}{\epsilon^2}\log\frac{1}{\delta})$ pricing queries and with probability $>1-\delta$ returns an estimate $\hat{q}$ for the quantile $Q(p)=1-F(p)$ such that $\abs{\hat q - Q(p)} < \epsilon$.

\end{lemma}


\subsection{Regular Distributions}\label{sec:regular}

Firstly we study the class of regular distributions $\reg$, supported on $[0,1]$. Regular distributions are the class of distributions for which the revenue curve in quantile space $\hat R(q) = q \cdot F^{-1}(1-q)$ is concave. This in particular implies that the revenue curve in the space of values/prices $\rev(p)=p(1-F(p))$  is single-peaked.
This class includes many important distributions such as uniform, exponential and all log-concave distributions.\\

In the heart of our algorithm will be the following lemma on regular distributions:

\begin{lemma}[Relative Flatness of Regular Distributions]\label{lem:myerson-average}
Let $\rev$ be the revenue curve of a regular distribution. 
Consider four equidistant\footnote{In the sense that $p_i = (p_4i + p_1(3-i))/3$ for $i=2,3$.} values $p_1<p_2<p_3<p_4=cp_1$ in $[0,1]$ and let $\rev_{\max}$ and $\rev_{\min}$ be the maximum and minimum value of the revenue curve at those four points. Then if $\rev_{\min} \geq \rev_{\max}-\eps$, for some $\eps>0$, then for any $p\in[p_1,p_4]$, $\rev(p)\leq \rev_{\max}+O(c\eps)$. 
\end{lemma}

The geometric intuition is that is the curve is reasonably flat at four equidistant points, then it should be reasonably flat within most of that interval. Note that this is not generally true for any single-peaked function, as the peak could be easily hiding between any of those points. See Figure \ref{fig:flatness}. The proof (deferred to Section~\ref{sec:proof-regular}) will strongly use the fact that the revenue curve is concave when looked at in the quantile space.  From this point on, the only two facts we will use about regular distributions is the fact it is single-peaked and the relative flatness lemma.






\begin{figure}[h]
\centering
\begin{tikzpicture}[scale=1]

\draw (0,0) -- (5,0);

\node[circle,fill,inner sep=1pt] at (1,.45) {};
\node[circle,fill,inner sep=1pt] at (2,.48) {};
\node[circle,fill,inner sep=1pt] at (3,.5) {};
\node[circle,fill,inner sep=1pt] at (4,.5) {};

\draw (1,-.1)--(1,.1);
\draw (2,-.1)--(2,.1);
\draw (3,-.1)--(3,.1);
\draw (4,-.1)--(4,.1);

\draw [line width=1pt, color=blue] plot [smooth] coordinates { (0,.2) (1,.45) (2,.48)  (2.2,1) (2.5,2) (2.8,1) (3,.5)  (4,.5) (5,.2)};
\end{tikzpicture}
\caption{The relative flatness lemma shows that the revenue curve of a regular distribution can't hide a very high revenue point in a region that appears flat when measured with respect to $4$ equidistant points. For example, the curve in the figure can't be the revenue curve of a regular distribution in the value space.}
\label{fig:flatness}
\end{figure}
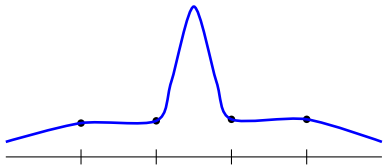

With that we can prove our main result:

\begin{theorem}\label{thm:regular}
Let $\reg$ be the class of regular distributions supported in $[0,1]$ and let $\theta$ be the monopoly price. Then,
$$\PrC_{\reg,\theta}(\epsilon) = \tilde{O}(1/\epsilon^2).$$
\end{theorem}

In the next section we will show this is tight even for the subclass of MHR distributions.\\

Now we describe the algorithm achieving that bound. The main data structure the algorithm will keep is a list  $L$ which will contain intervals with disjoint interiors in the revenue curve that  still need to be explored. For each of the endpoints of those intervals, we will keep an $\pm \epsilon$-estimate $\hat{q}(p)$ of its quantile $Q(p) = 1-F(p)$ obtained by the algorithm in Lemma \ref{lem:repeatprice} with success probability $\delta=\eps^2$ together with a revenue estimate:
$$\wrev(p) = p \cdot \hat{q}(p).$$

In the algorithm, we will show that we only price at $\tilde{O}(1)$ different prices, thus by union bound the overall success probability is at least $\frac{2}{3}$ for small $\eps$.

We initialize the algorithm with $L = \{[0,1]\}$. While the list is non-empty, we will process it by removing an interval from it and analyzing it. In the process of analyzing it, we will query some of its points. Depending on the the outcome of those queries, we may decide to add subintervals of $L$ back to the list for further analysis. We stop whenever the list is empty. At that point, we return the point queried in the process with largest estimated revenue.

We now describe how to process each interval. At every round, we remove the second leftmost interval $[\ell, r]$ from the list whenever the list has more than one interval. If not, we remove the only interval from the list. We process this interval by dividing it in $4$ pieces of the same length and querying the breakpoints
$p_i = (i \ell + (4-i) r)/4$ for $i = 1,2,3$
for the approximate quantile and construct a revenue estimate $\wrev$ for each. We define the following quantities for the interval:
$$\rev_{\max} = \max_{p \in \{\ell, p_1, p_2, p_3, r\}} \rev(p), \qquad 
\wrev_{\max} = \max_{p \in \{\ell, p_1, p_2, p_3, r\}} \wrev(p) $$

Now we proceed according to one of seven cases:\\

\noindent\emph{Case 1:} If $r - \ell < \epsilon$ we discard the interval.

In that case, we don't need to further analyze the interval since  no point
$p\in[\ell,r]$ can generate much better revenue than just pricing at $\ell$ since:

\[\rev(\ell)=\ell Q(\ell)\geq \ell Q(p)\geq (p-\eps)Q(p)\geq pQ(p)-\eps=\rev(p)-\eps.\]

\noindent\emph{Case 2:} If $\min_{p \in \{\ell, p_1, p_2, p_3, r\}} \wrev(p) \geq \wrev_{\max} - 2\epsilon$. In that case, we add $[\ell,p_1]$ to the list. At this point, since $r\leq 4p_1$, the relative flatness lemma guarantees us that no point in $[p_1, \ell]$ can be much better than the points that were already queried.\\

\noindent\emph{Case 3:} If $\min_{p \in \{p_1, p_2, p_3, r\}} \wrev(p) \geq \wrev_{\max} - 2\epsilon$ but $\wrev(\ell) <  \wrev_{\max} - 2\epsilon$ we remove all intervals to the left of $\ell$ from the list and add $[\ell, p_1]$ and $[p_1,r]$.

In this case, we know that:
$$\rev(\ell) \leq \wrev(\ell) + \epsilon < \wrev_{\max}-\epsilon \leq \rev_{\max}$$
and hence the revenue in $\ell$ is suboptimal and the optimal point $p^*$ is to the right of $\ell$. Here we are only using the fact that the revenue function is single-peaked.\\

All the remaining cases follow from similar arguments based on the single-peakedness of the revenue function:\\

\noindent\emph{Case 4:} If $\min_{p \in \{\ell, p_1, p_2, p_3\}} \wrev(p) \geq \wrev_{\max} - 2\epsilon$ but $\wrev(r) <  \wrev_{\max} - 2\epsilon$ we remove all intervals to the right of $r$ from the list and add $[\ell, p_3]$ and $[p_3,r]$.
\\

\noindent\emph{Case 5:} If $\min_{p \in \{p_1, p_2, p_3\}} \wrev(p) \geq \wrev_{\max} - 2\epsilon$ but $\wrev(\ell), \wrev(r) <  \wrev_{\max} - 2\epsilon$ then remove all intervals from the list and add back $[\ell, p_1]$, $[p_1, p_3]$ and $[p_3, r]$. \\

\noindent\emph{Case 6:} If $\wrev(\ell), \wrev(p_1)< \wrev_{\max} - 2\epsilon$ we  remove all intervals to the left of $\ell$ from the list and add $[p_1,r]$ to the list. \\

\noindent\emph{Case 7:} If $\wrev(r), \wrev(p_3) <\wrev_{\max} - 2\epsilon$ we  remove all intervals to the right of $r$ from the list and add $[\ell,p_3]$ to the list.\\

Observe that by single-peakedness those are the only possible cases.
The algorithm always finishes since it replaces intervals by other intervals that are smaller by at least a constant factor and once intervals become of size $\epsilon$, they are discarded. 
What remains to be done is to analyze the query complexity of the algorithm:

\begin{proof}[Proof of Theorem \ref{thm:regular}] Now we argue the algorithm described performs at most $\tilde{O}(1/\epsilon^2)$ queries. To see that, we first claim that there are at most three intervals in the candidate list at any given time. We prove this claim by induction over the number of iterations. Suppose that after the $t$-th iteration there are $k\leq 3$ intervals in $list$. We argue that after the $(t+1)$-th iteration, there are still at most $3$ intervals in the list.
Observe that in Case 1, 2, 6, 7, the number of intervals does not increase. In Case 5, the number of intervals become 3 no matter how many intervals were in the list. In Case 3 or 4, the number of intervals increase by at most 1, so the number of intervals in the list becomes at most $3$ if $k\leq 2$. The only remaining case is that the algorithm runs through Case 3 or 4, with $list$ containing $k=3$ intervals at the beginning of the iteration. Since the algorithm picks the second leftmost interval from $list$, the algorithm picks the middle one of the intervals (sorted according to the left boundary). Notice that in Case 3, the leftmost interval is also removed in this iteration; in Case 4, the rightmost interval is also removed in this iteration. Thus after the algorithm runs through Case 3 and Case 4, no matter how many intervals are in the list at the beginning of the iteration, at most 3 intervals remain in the list after the iteration. This finishes the proof of the claim.

After each iteration, either one of the interval of length $d$ splits to multiple intervals of length at most $\frac{3}{4}d$ each, with some other intervals getting removed; or one of the interval of length $d$ shrinks to a new interval of length at most $\frac{3}{4}d$ (recall that our five prices $\ell,p_1,p_2,p_3,r$ chosen in the algorithm are equidistant). Therefore if we record the lengths of the three intervals to be $d_1\geq d_2\geq d_3$, then either $d_1$ decreases to a fraction of at most $\frac{\sqrt{3}}{2}$; or $d_1$ does not increase, and $d_2$ decreases to a fraction of at most $\frac{\sqrt{3}}{2}$; or $d_1,d_2$ do not increase, and $d_3$ decreases to a fraction of at most $\frac{3}{4}$. For $i=1,2,3$, let $a_{i}=0$ if $d_i=0$,  and $a_{i}=\lceil\log_{2/\sqrt{3}}\frac{d_i}{\eps}\rceil$ if $d_i>0$. Then at each step, either $a_1$ decreases by at least 1; or $a_2$ decreases by at least 1 while $a_1$ does not increase; or $a_3$ decreases by at least 1. Let $n=1+\lceil\log_{2/\sqrt{3}}\frac{1}{\eps}\rceil$, and denote a potential function $\phi(d_1,d_2,d_3)=n^2a_1+na_2+a_3$. Then at each step, since $a_i\leq n-1$, $\phi$ decreases by at least 1. As at the beginning of the algorithm $\phi(d_1,d_2,d_3)=\phi(1,1,1)\leq n^3$, we have the total number of iterations is at most $n^3=O(\log^3\frac{1}{\eps})=\tilde{O}(1)$. Since in each iteration the query complexity is $\tilde{O}(\frac{1}{\eps^2})$, the total query complexity of the algorithm is $\tilde{O}(\frac{1}{\eps^2})$.
\end{proof}


\subsection{MHR Distributions}\label{sec:mhr}

We complement the above result by providing a lower bound on the pricing complexity for estimating the monopoly price for MHR distributions.

\begin{theorem}\label{thm:mhrlb}
Let $\mhr$ be the class of Monotone Hazard Rate distributions supported in $[0,1]$ and let $\theta$ be the monopoly price. Then
$$\PrC_{\mhr,\theta}(\epsilon) = \Omega(1/\epsilon^2).$$
\end{theorem}

To prove lower bound results for the pricing complexity for estimating the monopoly price, a key observation is that for two distributions with almost identical cumulative density functions, it's very hard to distinguish them with pricing queries. 

\begin{lemma}\label{lem:query-complexity}
For two value distributions $D$ and $D'$, if $\frac{1}{1+\eps}\leq\frac{Q_D(v)}{Q_{D'}(v)},\frac{F_D(v)}{F_{D'}(v)} \leq (1+\eps)$ for every $v\in[0,1]$, then $\Omega(\frac{1}{\eps^2})$ pricing queries are needed to distinguish the two distributions with probability $>1-\delta$.
\end{lemma}

The proof of the lemma is technically involved, and is deferred to Section~\ref{sec:proof-mhr}. We first show that the two distributions have small KL-divergence, then use Pinsker's inequality from information theory to bound the pricing complexity. A similar approach was also used in the literature \cite{HMR18}.

Thus to prove the lower bound on the pricing complexity of estimating the monopoly price for MHR distributions, we only need to find two MHR distributions with close cumulative density functions satisfying the above lemma. Such distributions indeed exist, and we defer the construction to Section~\ref{sec:proof-mhr}.

\subsection{General Distributions}\label{sec:general-distribution}

We now provide the pricing complexity for the class of general distributions.

\begin{theorem}\label{thm:general-distribution}
Let $\alldist$ be the class of all value distributions on $[0,1]$ and $\theta$ be the monopoly. Then $$\PrC_{\alldist,\theta}(\epsilon)= \tilde{\Theta}(1/\epsilon^3).$$
\end{theorem}

The proof idea is as follows. On one hand, we can price at all multiples of $\eps$ to get an estimated revenue for each price. We can use $\tilde{O}(\eps^{-2})$ pricing queries to get an $\eps$ accuracy at each price, thus $\tilde{O}(\eps^{-3})$ queries can guarantee us a price that estimates the optimal revenue with error $O(\eps)$. On the other hand, for a distribution with the same revenue on all but one multiples of $\eps$, suppose that pricing at $i\eps$ gives a higher revenue than any $j\eps$ with $j\neq i$. Any pricing algorithm needs to price at every multiple of $\eps$ to determine the exact value of $i$, while $\Omega(\eps^{-2})$ queries are needed for accuracy. Thus $\Omega(\eps^{-3})$ queries are necessary for the general class of distributions. The detailed proof is deferred to Section~\ref{sec:proof-general}.

\subsection{Discussion on Learning the Entire Regular Distribution}

The algorithms with optimal sample complexity designed by Huang, Mansour and Roughgarden \cite{HMR18} work by using the samples to build an empirical distribution and then choosing the optimal reserve price with respect to the empirical distribution. We know by the DKW inequality that after $\tilde{O}(1/\epsilon^2)$ samples, the empirical distribution will be $\epsilon$-close to the real distribution in Kolmogorov distance with high enough probability. Hence for a bounded distribution, there is no asymptotic gap between the sample complexity of learning the monopoly price and learning the entire CDF. Is there a similar phenomenon for pricing queries?

In Section~\ref{sec:regular}, we discussed how to \textit{directly} learn the optimal monopoly prices with $\eps$ error for regular distributions using $\tilde{O}(1/\eps^2)$ pricing queries. Using these many pricing queries, we may wonder if we can get a full picture of the distribution: is it possible to learn a distribution within a small distance to the true distribution? There are multiple metrics measuring the distance of two distributions. Metrics that are measured by the pointwise difference in cumulative density functions like Kolmogorov distance\footnote{Two distributions with cumulative density functions $F$ and $G$ are within $\eps$ Kolmogorov distance, if $F(x)-\eps\leq G(x)\leq F(x)+\eps$ for every $x\in\R$.} are inappropriate when we only have pricing access to the distribution. For example, if a distribution is defined by a deterministic value in $[0,1]$, pricing queries cannot determine that specific value. Instead, we measure the distance of two distributions by Levy distance\footnote{Two distributions with cumulative density functions $F$ and $G$ are within $\eps$ Levy distance, if $F(x-\eps)-\eps\leq G(x)\leq F(x+\eps)+\eps$ for every $x\in\R$.}. 

Observe that if two distributions $F$ and $G$ are within $\eps$ Levy distance, then for any price $p\in[0,1]$, $\rev_F(p-\eps)\geq \rev_G(p)-O(\eps)$. This means that if $p_F$ is the optimal monopoly price for $F$, then $p_F-\eps$ is a good price for $G$ that achieves near-optimal revenue (with $O(\eps)$ revenue loss). Thus, learning a distribution within a small Levy distance is a harder problem than learning the optimal monopoly price. For a general distribution in $[0,1]$, the pricing query complexity for learning an approximate distribution is $\tilde{\Theta}(\eps^{-3})$, which is the same as the complexity for learning the optimal monopoly price, as we can estimate the quantile of each multiple of $\eps$ using $\tilde{O}(\eps^{-2})$ pricing queries.  

What about regular distributions? Can we also find a regular distribution within $O(\epsilon)$ Levy distance using only $\tilde{O}(1/\eps^2)$ pricing queries? We answer the question negatively by showing that $\Omega(1/\eps^{2.5})$ pricing queries are required. This in particular shows that our zooming procedure based on the relative flatness lemma is necessary to obtain optimal pricing query complexity bounds. This also provides another scenario (other than learning the monopoly price for an MHR distribution) where there is an asymptotic gap between the pricing query complexity and the sample complexity (which is $O(1/\eps^2)$ by DKW inequality).

\begin{theorem}\label{thm:regular-levy-pricing}
Let $\reg$ be the class of regular distributions supported in $[0,1]$ and let $\theta$ be the true distribution\footnote{The distance metric $|\theta-\theta'|$ of two distributions $\theta$ and $\theta'$ is measured by Levy distance.}. Then 
$$\PrC_{\reg,\theta}(\epsilon) = \Omega(1/\epsilon^{2.5}).$$
\end{theorem}

\begin{proof}
Let $D^*$ be the uniform distribution on $[0,1]$. In other words, the density function $f$ of $D^*$ satisfies $f(v)=1$ for every $v\in[0,1]$. For any $x\in[0.1,0.9]$, let $D_x$ be the distribution with the following density function $f_x$:
\begin{equation*}
     f_x(v) = \left\{\begin{array}{lr}
        1, & \text{if } v\in[0,x)\cup(x+4\sqrt{\epsilon},1];\\
        v-x+1, & \text{if } v\in[x,x+\sqrt{\epsilon}];\\
        -v+x+1+2\sqrt{\epsilon}, & \text{if } v\in(x+\sqrt{\epsilon},x+3\sqrt{\epsilon});\\
        v-x-4\sqrt{\epsilon}+1, & \text{if } v\in[x+3\sqrt{\epsilon},x+4\sqrt{\epsilon}].\\
        \end{array}\right. 
\end{equation*}
Intuitively, $D_x$ is obtained by perturbing the uniform distribution $D^*$ in a small interval $[x,x+4\sqrt{\epsilon}]$. In that interval, $f_x(v)$ linearly increases from $1$ to $1+\sqrt{\epsilon}$, then decreases gradually to $1-\sqrt{\epsilon}$, finally back to $1$. 

Observe that distribution $D_x$ is regular since Myerson's virtual value $\phi_x(v)=v-\frac{1-F_x(v)}{f_x(v)}$ is increasing even when $f_x(v)$ decreases. Note: the derivative of $\phi_x$ is $\phi_x’(v)=2+\frac{(1-F_x(v))f_x’(v)}{f_x^2(v)}$. The density of $D_x$ satisfies $1-\sqrt{\epsilon}\leq f_x(v)\leq 1+\sqrt{\epsilon}$ and $-1\leq f_x’(v)\leq 1$, thus $\phi_x’(v)>0$ when $\epsilon$ is small enough, which means $\phi_x(v)$ is monotone increasing. Now we analyze the number of pricing queries needed for distinguishing $D^*$ and $D_x$.
 
The cumulative density function of $D_x$ is
\begin{equation*}
     F_x(v) = \left\{\begin{array}{lr}
        v, & \text{if } v\in[0,x)\cup(x+4\sqrt{\epsilon},1];\\
        v+\frac{(v-x)^2}{2}, & \text{if } v\in[x,x+\sqrt{\epsilon}];\\
        v+\epsilon-\frac{(x+2\sqrt{\epsilon}-v)^2}{2}, & \text{if } v\in(x+\sqrt{\epsilon},x+2\sqrt{\epsilon});\\
        v+\epsilon-\frac{(v-x-2\sqrt{\epsilon})^2}{2}, & \text{if } v\in[x+2\sqrt{\epsilon},x+3\sqrt{\epsilon});\\
        F(v)=v+\frac{(x+4\sqrt{\epsilon}-v)^2}{2}, & \text{if } v\in[x+3\sqrt{\epsilon},x+4\sqrt{\epsilon}].\\
        \end{array}\right. 
\end{equation*} 
Observe that for any $v\in[x,x+4\sqrt{\epsilon}]$, $v\leq F_x(v)\leq v+\epsilon$; for any $v\in[0,x)\cup(x+4\sqrt{\epsilon},1]$, $F_x(v)=v$. This means that the cdf $F_{x}$ of distribution $D_x$ differs from the cdf $F^*$ of uniform distribution $D^*$ by at most an additive $\epsilon$ in value range $[x,x+4\sqrt{\epsilon}]$; for $v$ being out of the range, $F_{x}(v)=F^*(v)=v$. Thus for $x\in[0.1,0.9]$, $\frac{F_{x}(v)}{F^*(v)}=1+O(\epsilon)$ for every $v\in[0,1]$. For quantile functions $Q_x(v)=1-F_x(v)$ and $Q^*(v)=1-F^*(v)$, similar results hold: $\frac{Q_{x}(v)}{Q^*(v)}=1-O(\epsilon)$ for every $v\in[0,1]$. Thus $\Omega(\epsilon^{-2})$ pricing queries are necessary due to Lemma~\ref{lem:query-complexity}. Furthermore, the queries must be in $[x,x+4\sqrt{\epsilon}]$, as $Q_{D_x}$ and $Q_{D^*}$ are identical outside of $[x,x+4\sqrt{\epsilon}]$, which means querying outside of $[x,x+4\sqrt{\epsilon}]$ gives no information.
 
For any $x\in[0.1,0.9]$, to distinguish $D_x$ and $D^*$, $\Omega(\epsilon^{-2})$ pricing queries in $[x,x+4\sqrt{\epsilon}]$ are necessary due to Lemma~\ref{lem:query-complexity}. Now consider the following $\Omega(\epsilon^{-0.5})$ distributions: $D^*$, $D_{0.1}$, $D_{0.1+4\sqrt{\epsilon}}$, $D_{0.1+8\sqrt{\epsilon}}$, $\cdots$, $D_{0.9-8\sqrt{\epsilon}}$, $D_{0.9-4\sqrt{\epsilon}}$. To distinguish all these distributions, $\Omega(\epsilon^{-2})$ queries are necessary in every interval $[x,x+4\sqrt{\epsilon}]$, thus $\Omega(\epsilon^{-2.5})$ queries in total are necessary. Notice that any two distributions $D_x$ and $D^*$ are at least $\frac{\epsilon}{2}$-far apart in Levy distance because, $F_x(x+2\sqrt{\epsilon})=x+2\sqrt{\epsilon}+\epsilon=F^*(x+2\sqrt{\epsilon}+\epsilon/2)+\epsilon/2$. To learn a distribution that is within $\frac{\epsilon}{4}$ Levy distance from the true distribution, we need to be able to distinguish all those distributions. To summarize, even in the class of regular distributions, $\Omega(\epsilon^{-2.5})$ pricing queries are necessary for learning a distribution within $\frac{\epsilon}{4}$ Levy distance (thus also $\epsilon$ Levy distance).

\end{proof}


\bibliographystyle{plainnat}
\bibliography{estimator}

\appendix

\section{Missing Proofs from Section \ref{sec:gaussian}}\label{sec:proof-gaussian}

\subsection{Missing proofs from Section~\ref{sec:median-general}}
\label{sec:proof-median-general}

\begin{customtheorem}{\ref{thm:medianub}}
Let $\C$ be a class of distributions over $\R$ with median in $[0,1]$. Define: $$a_F:=\sup_{F\in\C}\log{\textstyle|F^{-1}(\frac{1}{2}+\frac{\eps}{3})-F^{-1}(\frac{1}{2}-\frac{\eps}{3})|^{-1}}.$$ We have
\[\PrC_{\C,\theta}(\epsilon) \leq \tilde{O}(1/\epsilon^2)a_F\log a_F.\]
\end{customtheorem}

\begin{proof}[Proof of Theorem \ref{thm:medianub}]

Algorithm \ref{alg:binary_ucb} searches for a price $p^*$ with $F(p^*)\in[\frac{1}{2}-\eps,\frac{1}{2}+\eps]$. The algorithm starts with a range $[0,1]$ of the potential median, and repeatedly price at the middle point $p=\frac{1}{2}$ of the range.
After $n$ queries, if the current acceptance rate of the price $p$ is $\bar{q}$, then the true quantile $q=Q(p)=1-F(p)$ is bounded in a confidence interval $[\bar{q}-\tilde{O}(n^{-1/2}),\bar{q}+\tilde{O}(n^{-1/2})]$ with high probability by Chernoff bound. 
If $\frac{1}{2}$ is outside the confidence interval after $n$ steps, assume that $\frac{1}{2}>\bar{q}+\tilde{O}(n^{-1/2})$, then with high probability $q=Q(p)<\frac{1}{2}$, and $|\frac{1}{2}-q|=\Omega(n^{-1/2})$.
The algorithm proceeds with a new range $[\frac{1}{2},1]$, and repeat the binary search process until it finds a price $p$ with $\Omega(n^{-1/2})<\eps$.\\

Consider any iteration with price $p$ and corresponding quantile $q=Q(p) := 1 - F(p)$. Assume that $q<\frac{1}{2}$, the case where $q>\frac{1}{2}$ is similar. Let $\alpha=6\sqrt{\log(t/\delta)/n}$. By Chernoff bound, after $n$ pricing queries on $p$, the probability that the selling rate $\bar{q}=\frac{k}{n}\not\in[q-\alpha,q+\alpha]$ is
\begin{eqnarray*}
\Pr[\bar{q}\not\in[q-\alpha,q+\alpha]]\leq 2\exp\left(-\frac{1}{2}n\alpha^2\right)=\frac{2\delta^3}{t^3}<\frac{\delta}{2t^3}
\end{eqnarray*}
when $\delta\leq\frac{1}{2}$. Thus by union bound, the probability that at any time the true quantile $q\in[\bar{q}-\alpha,\bar{q}+\alpha]$ is $>1-\sum_t\frac{\delta}{t^2}>1-\delta$. Thus with probability $1-\delta$, the algorithm never falsely searches a range $[\ell,r]$ that does not contain $p^*$. 

When $\bar{q}$ is not out of the confidence interval, we have $|\frac{1}{2}-\bar{q}|\leq 2\alpha$, while $|\bar{q}-q|\leq \alpha$ with probability $1-\delta$. Thus with probability $1-\delta$, $\eps_p=3\alpha\geq |\frac{1}{2}-q|$. Therefore $\eps_p$ is indeed an upper bound of $|\frac{1}{2}-q|$.

Now we analyze the pricing complexity of the algorithm. When the algorithm prices at $p$ with $q=Q(p)$ such that $|q-\frac{1}{2}|\leq \frac{1}{3}\eps$, we show that the algorithm always returns with this price. Consider the following two cases. If the algorithm observes $|\frac{1}{2}-\bar{q}|>2\alpha$, then since with probability $1-\delta$ $|\bar{q}-q|\leq\alpha$, we have $|\frac{1}{2}-q|>\alpha$. Then $\eps_p=3\alpha<3|\frac{1}{2}-q|\leq \eps$, which means the algorithm returns the current price $p$. Otherwise the algorithm observes $\eps_p\leq \eps$ and returns. Thus if the algorithm sets some price $p$ with $q=Q(p)$ satisfying $|q-\frac{1}{2}|\leq \frac{1}{3}\eps$, then the algorithm would always return with this price. 

Since the algorithm always finds a price $p$ with $|Q(p)-\frac{1}{2}|\leq \frac{1}{3}\eps$ in $a_F=\log\frac{1}{F^{-1}(\frac{1}{2}+\frac{\eps}{3})-F^{-1}(\frac{1}{2}-\frac{\eps}{3})}$ rounds of different prices. Let $n$ be the number of queries in the last round of pricing. Then $t\leq a_Fn$, and $18\sqrt{\frac{\log(a_Fn/\delta)}{n}}<\eps$. Therefore $n=\tilde{O}(\frac{1}{\eps^2})\log\frac{a}{\delta}$, which means the query complexity of the algorithm is at most $a\tilde{O}(\frac{1}{\eps^2})\log\frac{a_F}{\delta}$.

\end{proof}

\subsection{Missing proofs from Section~\ref{sec:mean-normal}}
\label{sec:proof-mean-normal}

\begin{customcorollary}{\ref{cor:median-ub}}
Let $\bar\C_{\bar\sigma} = \left\{ \N(\mu, \sigma) \text{ s.t. } \mu \in [0,1],\ \sigma\leq\bar\sigma\right\}$ be the class of normal distributions with variance at most $\bar\sigma^2$ with the loss $\L(\theta, \hat \theta) = \abs{\theta - \hat \theta}$. Then the pricing complexity of computing the mean is:
$$\PrC_{\bar\C_{\bar\sigma},\mu}(\epsilon)= O(\bar\sigma^2/\epsilon^2).$$
\end{customcorollary}

\begin{proof}[Proof of Corollary~\ref{cor:median-ub}]
use the following variant of the algorithm for Theorem~\ref{lm:norm}:
\begin{itemize}
    \item Use Algorithm~\ref{alg:binary_ucb}  with $\epsilon = \frac{1}{12}$ to find  $p_1$ and $p_2$ with $Q(p_1) \in [1/6, 2/6]$ and $Q(p_2) \in [4/6, 5/6]$, 
    \item Using $O(\bar \sigma^2/\epsilon^2)$ pricing queries, find estimates $\hat q_i$ such that $\abs{\hat q_i - Q_{\mu,\sigma}(p_i)} \leq \epsilon/\bar \sigma$. 
    \item Solve the system of equations $\hat{q}_1 = Q_{\hat \mu, \hat \sigma}(p_1)$ and $\hat{q}_2 = Q_{\hat \mu, \hat \sigma}(p_2)$ to find estimates of $\hat \mu, \hat \sigma$.
\end{itemize}
By the same argument as in the proof of Theorem~\ref{lm:norm}, the equations above can be re-written in linear form:
$$\hat \mu + \hat \sigma \cdot Q_{0,1}^{-1}(\hat q_1) = p_1 \qquad \hat \mu + \hat \sigma \cdot Q_{0,1}^{-1}(\hat q_2) = p_2$$
The actual mean and variance are solutions to the same equation with exact quantiles:
$$ \mu + \sigma \cdot Q_{0,1}^{-1}( q_1) = p_1 \qquad  \mu +  \sigma \cdot Q_{0,1}^{-1}( q_2) = p_2$$
Comparing those equations we can observe the following (we omit the details of the calculations as they are quite standard):
$$\abs{\mu - \hat \mu} \leq O\left( \sigma \abs{p_1 - p_2} \cdot \max_i \abs{Q_{0,1}^{-1}(\hat q_i) - q_i} \right) = O\left(\sigma \cdot \frac{\epsilon}{\bar \sigma}\right) = O(\epsilon)$$
This leads to an algorithm with $O(\bar \sigma^2 / \epsilon^2)$ pricing queries that only requires knowledge of a \emph{variance upper bound}.

\end{proof}

\section {Missing Proofs from Section \ref{sec:monopoly}}
\label{sec:proof-monopoly}


\begin{customlemma}{\ref{lem:repeatprice}}
For any value distribution $F$ and value $p$, let $Q(p)=1-F(p)$ be the probability that a random sample from $F$ is at least $p$. Then if we make $m=\tilde{O}(\frac{1}{\epsilon^2}\log\frac{1}{\delta})$ pricing queries with price $p$, then with probability $>1-\delta$ the price is accepted $m(Q(p)\pm\epsilon)$ times. 
\end{customlemma}
\begin{proof}[Proof of Lemma~\ref{lem:repeatprice}]
Let $q^*=Q(p)$ be the probability that $v\sim F$ is at least $p$. For any $m$ queries of price $p$, suppose that there are $t$ queries with signal ``$v\geq p$''. By Chernoff bound,
\begin{eqnarray*}
\Pr\left(t\not\in[m(q^*-\eps)]\right)\leq 2\exp\left(-\frac{\eps^2m}{2q^*}\right)\leq 2\exp\left(-\frac{\eps^2m}{2}\right).
\end{eqnarray*}
When $m=2\frac{1}{\eps^2}\log\frac{2}{\delta}=\Theta(\frac{1}{\eps^2}\log\frac{1}{\delta})$, the right hand side of the above inequality be at most $\delta$. Thus if $q^*\not\in[q-\eps,q+\eps]$ for some probability $q$, after $m=\Theta(\frac{1}{\eps^2}\log\frac{1}{\delta})$ pricing queries on $p$, with probability $<\delta$ the number of queries with signal ``$v\geq p$'' is in $[m(q^*-\eps),m(q^*+\eps)]$.
\end{proof}

\subsection{Missing proofs from Section~\ref{sec:regular}}
\label{sec:proof-regular}

\begin{customlemma}{\ref{lem:myerson-average}}[Relative Flatness of Regular Distributions]
Let $\rev$ be the revenue curve of a regular distribution. 
Consider four equidistant values $p_1<p_2<p_3<p_4=cp_1$ in $[0,1]$ and let $\rev_{\max}$ and $\rev_{\min}$ be the maximum and minimum value of the revenue curve at those four points. Then if $\rev_{\min} \geq \rev_{\max}-\eps$, for some $\eps>0$, then for any $p\in[p_1,p_4]$, $\rev(p)\leq \rev_{\max}+O(c\eps)$. 
\end{customlemma}
\begin{proof}[Proof of Lemma~\ref{lem:myerson-average}]
For each $i=1,2,3,4$, let $q_i=Q(p_i)$. Denote $q=Q(p)$. Since $p_1-\ell=p_2-p_1=p_3-p_2=p_4-p_3$, we have $p_1<p_4=cp_1$. 
If $\rev_{\max}<2\eps$, then for any $p\in[p_1,p_4]$, $\rev(p)=pQ(p)\leq p_4Q(p_1)=cp_1Q(p_1)\leq c\rev_{\max}=O(c\eps)$, thus $\rev(p)\leq\rev_{\max}+O(c\eps)$. 
Therefore we assume $\rev_{\max}\geq 2\eps$, then $\rev(p_1),\rev(p_2),\rev(p_3),\rev(p_4)\geq \frac{1}{2}\rev_{\max}$. Thus $q_1=\frac{1}{p_1}\rev(p_1)\leq \frac{c}{p_4}\cdot2\rev(p_4)=2cq_4$, in other words, all $q\in[q_4,q_1]$ are within a constant factor of $2c$.

Since $\rev$ is the revenue curve of a regular distribution, it is a concave function in quantile space. Consider the following two cases:

\noindent\emph{Case 1:} $p\geq Q^{-1}(\frac{q_2+q_3}{2})$. 
We can express $q_2$ as a linear combination of $q_1$ and $q$ by
$q_2=\frac{q_2-q}{q_1-q}q+\frac{q_1-q_2}{q_1-q}q_1$, thus
\[\rev(p_2)\geq \frac{q_2-q}{q_1-q}\rev(p)+\frac{q_1-q_2}{q_1-q}\rev(p_1).\]
Then
\begin{eqnarray}
\rev(p)&\leq&\frac{(q_1-q)\rev(p_2)-(q_1-q_2)\rev(p_1)}{q_2-q}\leq\frac{(q_1-q)(\rev(p_1)+\eps)-(q_1-q_2)\rev(p_1)}{q_2-q}\nonumber\\
&=&\rev(p_1)+\frac{q_1-q}{q_2-q}\eps=\rev(p_1)+\frac{q_1-q_2}{q_2-q}\eps+\eps\leq\rev(p_1)+\frac{q_1-q_2}{q_2-q_3}2\eps+\eps.\label{eqn:revp}
\end{eqnarray}

Here the last inequality is by $q\leq\frac{q_2+q_3}{2}$.

If $\eps>\frac{1}{2}q_3(p_3-p_2)$, then 
\[\eps>\frac{1}{2}q_3(p_3-p_2)\geq \frac{1}{2}\cdot\frac{1}{2c}q\cdot\frac{1}{2}(p-p_2)\geq\frac{1}{8c}(qp-q_2p_2)=\frac{1}{8c}(\rev(p)-\rev(p_2)).\]
Here the second inequality is by $q\geq\frac{1}{2c}q_3$, and $p-p_2\leq p_4-p_2=2(p_3-p_2)$; the third inequality is by $q\leq q_2$. 
Therefore $\rev(p)\leq\rev(p_2)+O(c\eps)=\rev_{\max}+O(c\eps)$ if $\eps>\frac{1}{2}q_3(p_3-p_2)$.

Now we assume that  $\eps\leq\frac{1}{2}q_3(p_3-p_2)$.
Since $\rev(p_1)\leq \rev(p_2)+\eps$, we have $p_1q_1\leq p_2q_2+\eps$, then $q_1-q_2\leq\frac{q_1(p_2-p_1)+\eps}{p_1}$. At the same time, $\rev(p_2)\geq\rev(p_3)-\eps$, we have $p_2q_2\geq p_3q_3-\eps$, then $q_2-q_3\geq\frac{q_3(p_3-p_2)-\eps}{p_2}$. Thus
\begin{eqnarray*}
\frac{q_1-q_2}{q_2-q_3}&\leq& \frac{p_2(q_1(p_2-p_1)+\eps)}{p_1(q_3(p_3-p_2)-\eps)}\leq \frac{p_2(q_1(p_2-p_1)+\frac{1}{2}q_3(p_3-p_2))}{p_1(q_3(p_3-p_2)-\frac{1}{2}q_3(p_3-p_2))}\\
&\leq&\frac{2(q_1+\frac{1}{2}q_3)}{(q_3-\frac{1}{2}q_3)}=2+\frac{4q_1}{q_3}\leq 2+8c=O(c).
\end{eqnarray*}
Here the second inequality is by $\eps\leq\frac{1}{2}q_3(p_3-p_2)$; the third inequality is by $p_2\leq 2p_1$ and $p_3-p_2=p_2-p_1$; the last inequality is by $q_1\leq 2cq_3$.
By \eqref{eqn:revp} we have $\rev(p)\leq \rev_{\max}+O(c\eps)$.

\noindent\emph{Case 2:} $p< Q^{-1}(\frac{q_2+q_3}{2})$. We can express $q_3$ as a linear combination of $q$ and $q_4$ by $q_3=\frac{q-q_3}{q-q_4}q+\frac{q_3-p_4}{q-q_4}q_4$. Thus by concavity of $\rev$ in the quantile space,
\[\rev(p_3)\geq \frac{q-q_3}{q-q_4}\rev(p)+\frac{q_3-q_4}{q-q_4}\rev(p_4).\]
Then
\begin{eqnarray}
\rev(p)&\leq&\frac{(q-q_4)\rev(p_3)-(q_3-q_4)\rev(p_4)}{q-q_3}\leq\frac{(q-q_4)(\rev(p_4)+\eps)-(q_3-q_4)\rev(p_4)}{q-q_3}\nonumber\\
&=&\rev(p_4)+\frac{q-q_4}{q-q_3}\eps=\rev(p_4)+\frac{q_3-q_4}{q-q_3}\eps+\eps<\rev(p_4)+\frac{q_3-q_4}{q_2-q_3}2\eps+\eps.\label{eqn:revp-case2}
\end{eqnarray}
Here the last inequality is by $q>\frac{q_2+q_3}{2}$.

If $\eps>\frac{1}{2}q_3(p_3-p_2)$, then as we have reasoned in Case 1, $\rev(p)\leq \rev_{\max}+O(c\eps)$. Now we assume that $\eps\geq\frac{1}{2}q_3(p_3-p_2)$. Since $\rev(p_3)\leq \rev(p_4)+\eps$, we have $p_3q_3\leq p_4q_4+\eps$, then $q_3-q_4\leq\frac{(p_4-p_3)q_4+\eps}{p_3}$. At the same time, as we have shown in Case 1, $q_2-q_3\geq\frac{q_1(p_3-p_2)-\eps}{p_2}$. Thus
\begin{eqnarray*}
\frac{q_3-q_4}{q_2-q_3}&\leq& \frac{p_2(q_4(p_4-p_3)+\eps)}{p_3(q_3(p_3-p_2)-\eps)}\leq \frac{p_2(q_4(p_4-p_3)+\frac{1}{2}q_3(p_3-p_2))}{p_3(q_3(p_3-p_2)-\frac{1}{2}q_3(p_3-p_2))}\\
&\leq&\frac{q_4+\frac{1}{2}q_3}{q_3-\frac{1}{2}q_3}=1+\frac{2q_4}{q_3}\leq 1+2=3.
\end{eqnarray*}
Here the second inequality is by $\eps\leq\frac{1}{2}q_3(p_3-p_2)$; the third inequality is by $p_2\leq p_3$ and $p_3-p_2=p_4-p_3$; the last inequality is by $q_4\leq q_3$.
By \eqref{eqn:revp-case2} we have $\rev(p)\leq \rev_{\max}+O(c\eps)$.

\end{proof}

\subsection{Mission proofs from Section~\ref{sec:mhr}}
\label{sec:proof-mhr}
In this section, we prove the pricing query complexity for estimating the monopoly price of MHR distributions.
\begin{customtheorem}{\ref{thm:mhrlb}}
Let $\mhr$ be the class of Monotone Hazard Rate distributions supported in $[0,1]$ and let $\theta$ be the monopoly price. Then
$$\PrC_{\mhr,\theta}(\epsilon) = \Omega(1/\epsilon^2).$$
\end{customtheorem}

Before proving the theorem, we start with a technical lemma bounding the KL-divergence of two Bernoulli variables.

\begin{lemma}\label{lem:kl-bernoulli}
For any $\eps<\frac{\sqrt{2}}{2}$ and two Bernoulli variables $X,Y$ with $\frac{1}{1+\eps}\leq \frac{\Pr[X=1]}{\Pr[Y=1]},\frac{\Pr[X=0]}{\Pr[Y=0]}\leq 1+\eps$, $\dkl(X\|Y)<\eps^2$.
\end{lemma}
\begin{proof}[Proof of Lemma~\ref{lem:kl-bernoulli}]
Let $q=\Pr[X=1]$, $q'=\Pr[Y=1]$. Without loss of generality assume $q\leq \frac{1}{2}$ (otherwise just set $q\leftarrow 1-q$ and $q'\leftarrow 1-q'$). Let $f(q,q')=\dkl(X\|Y)$. Then
\begin{equation*}
    f(q,q')=\dkl(X\|Y)=q\ln\frac{q}{q'}+(1-q)\ln\frac{1-q}{1-q'}=q\ln q-q\ln q'+(1-q)\ln(1-q)-(1-q)\ln (1-q').
\end{equation*}
Then
\begin{equation*}
    \frac{\partial}{\partial q}f(q,q')=\ln q-\ln q'-\ln(1-q)+\ln(1-q')=\ln\frac{q}{q'}-\ln\frac{1-q}{1-q'}
\end{equation*}
has unique zero point $q=q'$, which means $f(q,q')$ increases when $q>q'$, and decreases when $q<q'$.
Therefore, to upper bound $f(q,q')$, it suffices to bound it for $q=(1+\eps)q'$ and $q=(1-\eps)q'$.

\paragraph{Case 1.} When $q=(1+\eps)q'$, 
\begin{eqnarray*}
f(q,q')&=&q\ln\frac{q}{q'}+(1-q)\ln\frac{1-q}{1-q'}\\
&=&q\ln(1+\eps)+(1-q)\ln(1-q)-(1-q)\ln\left(1-\frac{q}{1+\eps}\right).
\end{eqnarray*}
Take the derivative of the right hand side with respect to $q$, we have
\begin{eqnarray*}
\frac{\partial}{\partial q}f\left(q,\frac{q}{1+\eps}\right)&=&\ln(1+\eps)-1-\ln(1-q)+\ln\left(1-\frac{q}{1+\eps}\right)+(1-q)\frac{1}{1+\eps}\frac{1}{1-\frac{q}{1+\eps}}\\
&=&(\ln(1+\eps)-1)+\ln\frac{1-\frac{q}{1+\eps}}{1-q}+\frac{1-q}{1+\eps-q}>0.
\end{eqnarray*}
The inequality is true since each term in the sum is non-negative. Thus $f(q,\frac{q}{1+\eps})$ is maximized when $q=\frac{1}{2}$, and
\begin{eqnarray*}
f\left(q,\frac{q}{1+\eps}\right)\leq f\left(\frac{1}{2},\frac{1}{2+2\eps}\right)
=\frac{1}{2}\ln(1+\eps)+\frac{1}{2}\ln\frac{\frac{1}{2}}{1-\frac{1}{2+2\eps}}=\frac{1}{2}\ln\frac{(1+\eps)^2}{1+2\eps}<\frac{1}{2}\frac{\eps^2}{1+2\eps}<\eps^2,
\end{eqnarray*}
here the second inequality is by $\ln(1+x)<x$ for any $x>0$.

\paragraph{Case 2.} When $q=\frac{1}{1+\eps}q'$, 
\begin{eqnarray*}
f(q,q')&=&q\ln\frac{q}{q'}+(1-q)\ln\frac{1-q}{1-q'}\\
&=&-q\ln(1+\eps)+(1-q)\ln(1-q)-(1-q)\ln\left(1-(1+\eps)q\right).
\end{eqnarray*}
Take the derivative of the right hand side with respect to $q$, we have
\begin{eqnarray*}
\frac{\partial}{\partial q}f\left(q,(1+\eps)q\right)&=&-\ln(1+\eps)-1-\ln(1-q)+\ln\left(1-(1+\eps)q\right)+(1-q)(1+\eps)\frac{1}{1-(1+\eps)q}\\
&=&-\ln\frac{(1-q)(1+\eps)}{1-(1+\eps)q}-1+\frac{(1-q)(1+\eps)}{1-(1+\eps)q}\geq0.
\end{eqnarray*}
The inequality follows from $-\ln x-1+x\geq0$ for any $x\geq 0$. Thus $f\left(q,(1+\eps)q\right)$ is maximized when $q=\frac{1}{2}$, and 
\begin{eqnarray*}
f\left(q,(1+\eps)q\right)\leq f\left(\frac{1}{2},\frac{1+\eps}{2}\right)
=-\frac{1}{2}\ln(1+\eps)+\frac{1}{2}\ln\frac{\frac{1}{2}}{1-\frac{1}{2}(1+\eps)}=\frac{1}{2}\ln\frac{1}{1-\eps^2}<\frac{1}{2}\frac{\eps^2}{1-\eps^2}< x^2,
\end{eqnarray*}
here the second inequality is by $\ln(1+x)<x$ for any $x>0$, and the last inequality is by $\eps>\frac{\sqrt{2}}{2}$.

\end{proof}

The lemma is then used to give a lower bound on the pricing query complexity of distinguishing two distributions whose c.d.f. $F(p)$ and quantiles $Q(p) = 1-F(p)$ are both within $1\pm \epsilon$ of each other.

\begin{customlemma}{\ref{lem:query-complexity}}
For two value distributions $D$ and $D'$, if $\frac{1}{1+\eps}\leq\frac{Q_D(v)}{Q_{D'}(v)},\frac{F_D(v)}{F_{D'}(v)} \leq (1+\eps)$ for every $v\in[0,1]$, then $\Omega(\frac{1}{\eps^2})$ pricing queries are needed to distinguish the two distributions with probability $>1-\delta$.
\end{customlemma}

The proof idea of Lemma~\ref{lem:query-complexity} is as follows. For two such distributions $D$ and $D'$, and any pricing algorithm $\A$, let $X_i$ and $X'_i$ be the bit each algorithm receives from the $i$-th pricing query. To show that $\A$ needs $m=\Omega(\eps^{-2})$ queries to distinguish the two distribution, it suffices to show that $(X_1,\cdots,X_m)$ and $(X'_1,\cdots,X_m)$ has $\Omega(1)$ statistical distance (thus KL-divergence $\Omega(1)$) only if $\Omega(\eps^{-2})$. By Lemma~\ref{lem:kl-bernoulli} the gain of relative entropy from each query is at most $\eps^2$, thus $\Omega(\eps^{-2})$ queries are needed for the KL-divergence to become $\Omega(1)$. 

\begin{proof}[Proof of Lemma~\ref{lem:query-complexity}]
Consider any algorithm $\A$ that adaptively sets a pricing query in each step. 
For every $i\geq 1$, let $X_{i}$ be a boolean random variable that denotes whether the $i$-th sampled value $v_i\sim D$ is at least the price $p_{i,D}$ at the $i$-th step in algorithm $\A$ for value distribution $D$.
In other words, $X_i=\one[v_i\geq p_{i,D}]$. Similarly, define $X'_{i}$ to be a boolean random variable that denotes whether the $i$-th sampled value $v'_i\sim D'$ is at least the price $p_{i,D'}$ at the $i$-th step in algorithm $\A$ for value distribution $D'$. For each $n\geq 1$, denote $\X_{\leq n}=(X_1,X_2,\cdots,X_n)$ and $\X'_{\leq n}=(X'_1,X'_2,\cdots,X'_n)$.

For any price $p\in[0,1]$, random values $v\sim D$ and $v'\sim D'$, and random variables $X=\one[v\geq p]$ and $Y=\one[v'\geq p]$, let $q=\Pr[X=1]$ and $q'=\Pr[Y=1]$. The Kullback–Leibler divergence of $X$ and $Y$ can be bounded by Lemma \ref{lem:kl-bernoulli}.

Then for any fixed $(x_1,\cdots,x_{m-1})\in\{0,1\}^{m-1}$, conditioned on $(X_1,\cdots,X_{m-1})=(x_1,\cdots,x_{m-1})$ algorithm $\A$ sets a random price $p$; $X_{m}$ is a Bernoulli variable with probability $q=\E_{p}[\Pr_{v\sim D}[v\geq p]]=\E_{p}[Q_D(p)]$ being 1; $X'_{m}$ is a Bernoulli variable with probability $q'=\E_{p}[\Pr_{v\sim D'}[v\geq p]]=\E_{p}[Q_{D'}(p)]$ being 1. By the condition of Lemma~\ref{lem:query-complexity}, $\frac{1}{1+\eps}<\frac{q}{q'},\frac{1-q}{1-q'}<1+\eps$. Thus by Lemma~\ref{lem:kl-bernoulli},
\begin{equation*}
    \dkl\bigg(X_m|_{\X_{\leq m-1}=(x_1,\cdots,x_{m-1})}\bigg\|X'_m|_{\X'_{\leq m-1}=(x_1,\cdots,x_{m-1})}\bigg)<\eps^2.
\end{equation*}
For any $n\geq 1$, let $\x_{\leq n}=(x_1,x_2,\cdots,x_n)$. Then by the chain rule of KL divergence,
\begin{eqnarray*}
\dkl(X_{\leq m}\|X'_{\leq m})
&=&\dkl(X_{\leq m-1}\|X'_{\leq m-1})+\E_{\X_{\leq m-1}}\dkl\bigg(X_m|_{\X_{\leq m-1}=\x_{\leq m-1}}\bigg\|X'_m|_{\X'_{\leq m-1}=\x_{\leq m-1}}\bigg)\\
&<&\dkl(X_{\leq m-1}\|X'_{\leq m-1})+\eps^2<\dkl(X_{\leq m-2}\|X'_{\leq m-2})+2\eps^2<\cdots<m\eps^2.
\end{eqnarray*}


Let $D_m$ and $D'_m$ be the distribution of $\X_{\leq m}$ and $\X'_{\leq m}$ respectively. The probability that algorithm $\A$ can identify the distribution with $m$ samples is at most $\frac{1+\delta(D_m,D'_m)}{2}$, here $\delta(A,B)=\frac{1}{2}\int |f_A(x)-f_B(x)|dx$ represents the statistical distance of two distributions $A$ and $B$ with density $f_A$ and $f_B$ respectively. Thus to identify distribution $D$ and $D'$ with probability $1-\delta$, the number of samples $m$ need to be large enough such that $\delta(D_m,D'_m)=\Omega(1)$. By Pinsker's inequality,
\begin{equation*}
    \delta(D_m,D'_m)\leq\sqrt{\frac{1}{2}\dkl(D_m,D'_m)}<\sqrt{\frac{m}{2}\eps^2}.
\end{equation*}
Thus when $\delta(D_m,D'_m)=\Omega(1)$, $\sqrt{\frac{m}{2}\eps^2}=\Omega(1)$. Thus $m=\Omega(\frac{1}{\eps^2})$ samples are necessary to distinguish $D$ and $D'$ with probability $>\frac{2}{3}$ for any (randomized adaptive) algorithm $\A$.
\end{proof}

Finally, we are ready to analyze the pricing query complexity for estimating the monopoly price for MHR distributions.

\begin{proof} [Proof of Theorem~\ref{thm:mhrlb}]
We consider the following MHR distribution $D$ with sample complexity $\Omega(\eps^{-3/2})$ that is modified from an example in Huang et al \cite{HMR18}. Let $\eps_0=16\eps$, and $f$ be the density of $D$ as follows:
\begin{equation*}
    f_D(v)=\begin{cases}
        0.4 &\textrm{ if } 0\leq v< \frac{1}{2};\\
        1.6-3.2\sqrt{\eps_0}&\textrm{ if } \frac{1}{2}\leq v< \frac{1}{2}+\frac{1}{2}\sqrt{\eps_0};\\
        1.6+\frac{3.2\eps_0}{1-\sqrt{\eps_0}}&\textrm{ if } \frac{1}{2}+\frac{1}{2}\sqrt{\eps_0}\leq v\leq 1.
    \end{cases}
\end{equation*}
Let $f_{D'}(v)$ be the density of distribution $D'$ as follows:
\begin{equation*}
    f_{D'}(v)=\begin{cases}
        0.4 & \textrm{ if } 0\leq v< \frac{1}{2};\\
        1.6 & \textrm{ if } \frac{1}{2}\leq v\leq 1.
    \end{cases}
\end{equation*}
In other words, distribution $D'$ is uniform on $[0,\frac{1}{2}]$ and $[\frac{1}{2},1]$, while distribution $D$ is a perturbation of $D'$ on $[\frac{1}{2},1]$, with the density of $D'$ in $[\frac{1}{2},\frac{1}{2}+\frac{1}{2}\sqrt{\eps_0}]$ scaled down by a factor of $1-2\sqrt{\eps_0}$, and the density of $D'$ in $[\frac{1}{2}+\frac{1}{2}\sqrt{\eps_0},1]$ scaled up by a factor of $1+\frac{2\eps_0}{1-\sqrt{\eps_0}}$ in $D$. Let $F_D$ and $F_{D'}$ be the cumulative density function of $D$ and $D'$ respectively, and let $Q_D(v)=1-F_D(v)$, $Q_{D'}(v)=1-F_{D'}(v)$. Then for any $v\in[0,1]$, $Q_{D'}(v)\leq Q_{D}(v)\leq \left(1+\frac{2\eps_0}{1-\sqrt{\eps_0}}\right)Q_{D'}(v)$. In other words, when setting any pricing query over the two distributions, the probability that the buyer can afford to purchase the item differs by only a factor of $1+\frac{2\eps_0}{1-\sqrt{\eps_0}}<1+3\eps_0$ for small enough $\eps_0$.

For uniform distribution $D'$, the optimal monopoly price is $p_{D'}^*=\frac{1}{2}$, with revenue $0.4$. Only prices in $p\leq p_0=\frac{1}{2}+\sqrt{\eps}$ would lead to revenue $\geq0.4-1.6\eps$. For uniform distribution $D$, the optimal monopoly price is $p_D^*=\frac{1}{2}+\frac{1}{2}\sqrt{\eps_0}$, with revenue $0.4+0.4\eps_0+0.8\eps_0^{3/2}=0.4+6.4\eps+51.2\eps^{3/2}>p^*_{D'}Q_{D'}(p^*_{D'})+\eps$. For $p\leq p_0$, the revenue is
\begin{eqnarray*}
pQ_D(p)&\leq& p_0Q_D(p_0)=(0.5+\sqrt{\eps})(0.8-1.6\sqrt{\eps}(1-2\sqrt{\eps_0}))
=0.4-1.6\eps+1.6\sqrt{\eps\eps_0}+3.2\eps\sqrt{\eps_0}\\
&=&0.4+4.8\eps+12.8\eps^{3/2}<p^*_{D}Q_D(p^*_D)-\eps.
\end{eqnarray*}
Thus no price can get $\eps$-close to the optimal revenue for both distribution $D$ and $D'$ at the same time. 

Consider a distribution that is one of $D$ and $D'$. To learn the optimal monopoly price or the optimal revenue of the distribution, one must be able to distinguish the two distributions with pricing queries. Observe that for any $v\leq\frac{1}{2}$, $F_D(v)=F_{D'}(v)$ and $Q_D(v)=Q_{D'}(v)$, thus $\frac{F_D(v)}{F_{D'}(v)}=\frac{Q_D(v)}{Q_{D'}(v)}=1$. For any $v\in(\frac{1}{2},1]$, 
\begin{equation*}
    1\geq \frac{F_D(v)}{F_{D'}(v)}\geq\frac{F_D(\frac{1}{2}+\frac{1}{2}\sqrt{v_0})}{F_{D'}(\frac{1}{2}+\frac{1}{2}\sqrt{v_0})}=1-\frac{1.6\eps_0}{0.2+0.8\sqrt{\eps_0}}=\frac{1}{1+O(\eps)}
\end{equation*}
and 
\begin{equation*}
    1\leq \frac{Q_D(v)}{Q_{D'}(v)}\leq\frac{Q_D(\frac{1}{2}+\frac{1}{2}\sqrt{v_0})}{Q_{D'}(\frac{1}{2}+\frac{1}{2}\sqrt{v_0})}=  1+\frac{1.6\eps_0}{0.8-0.8\sqrt{\eps_0}}=1+O(\eps)
\end{equation*}
by $\eps_0=16\eps$. Thus for any $v\in[0,1]$, $\frac{1}{1+O(\eps)}\leq \frac{F_D(v)}{F_{D'}(v)},\frac{Q_D(v)}{Q_{D'}(v)}\leq 1+O(\eps)$. By Lemma~\ref{lem:query-complexity} the number of queries needed to distinguish $D$ and $D'$ is $\Omega(\frac{1}{\eps^2})$, which is also a lower bound on the pricing complexity of estimating the monopoly price for MHR distributions.

\end{proof}

\subsection{Missing proofs from Section~\ref{sec:general-distribution}}\label{sec:proof-general}

\begin{customtheorem}{\ref{thm:general-distribution}}
Let $\alldist$ be the class of all value distributions on $[0,1]$ and $\theta$ be the monopoly. Then $$\PrC_{\alldist,\theta}(\epsilon)= \tilde{\Theta}(1/\epsilon^3).$$
\end{customtheorem}

\begin{proof}[Proof of Theorem~\ref{thm:general-distribution}]
Firstly we show that $\PrC_{\alldist,\theta}(\epsilon)= \tilde{O}(1/\epsilon^3).$ Actually, consider the following simple algorithm: price at every multiple of $\eps$ in $[0,1]$ for $\tilde{O}(\frac{1}{\eps^2})$ times such that for every such price $p$, $Q(p)$ is estimated by $\tilde{Q}(p)\in Q(p)\pm\eps$ (by Lemma~\ref{lem:repeatprice}). Notice that for price $p^*$ that maximizes the revenue $\rev(p^*)=p^*Q(p^*)$, the revenue from price $p=\eps\lfloor\frac{p^*}{\eps}\rfloor$ is $\rev(p)=pQ(p)\geq (p^*-\eps)Q(p^*)\geq \rev(p^*)-\eps$. Since the algorithm estimate $Q(p)$ with additive error $\eps$, thus $\wrev(p)\geq \rev(p)-\eps\geq \rev(p^*)-2\eps$. This means that the price $\hat{p}$ being a multiple of $\eps$ that maximizes $\hat{p}\tilde{Q}(\hat{p})$ estimates the optimal revenue with error $O(\eps)$.

Now we show that the $\Omega(\frac{1}{\eps^3})$ pricing complexity is unavoidable. Consider the following $m=\frac{1}{16\eps}$ distributions $F_0, F_1,\cdots,F_{m}$. Distribution $F_i$ has support $\frac{1}{2}+4\eps$, $\frac{1}{2}+8\eps$, $\cdots$, $\frac{3}{4}-4\eps$, $\frac{3}{4}$, with the item-pricing revenue and quantiles satisfying
\begin{equation*}
    \rev_i\left(\frac{1}{2}+4k\eps\right)=\left(\frac{1}{2}+4k\right) \cdot Q_i\left(\frac{1}{2}+4k\right)=\begin{cases}
    \frac{1}{4},&\textrm{ if }k\neq i;\\
    \frac{1}{4}+\eps,&\textrm{ if }k=i.
    \end{cases}
\end{equation*}
In other words, for $i\geq 1$, each distribution $F_i$ has a unique revenue-maximizing price $\frac{1}{2}+4i\eps$ with revenue $\frac{1}{4}+\eps$, while other prices leads to revenue $\frac{1}{4}$ that is $\eps$-far from the optimal revenue. $F_0$ is an equal-revenue distribution with revenue $\frac{1}{4}$ for every price $\frac{1}{2}+4k\eps$. Thus the quantile of any two distributions $F_0$ and $F_{i}$ only differs at one point: $Q_i(\frac{1}{2}+4i\eps)=Q_{0}(\frac{1}{2}+4i\eps)+\Theta(\eps)$. For any other price $p=\frac{1}{2}+4k\eps$ with $k\neq i,j$, $Q_i(p)=Q_0(p)$. Then for any $i\in[\frac{m}{3},\frac{2m}{3}]$, since $Q_i(\frac{1}{2}+4i\eps)=\Theta(1)$ and $F_i(\frac{1}{2}+4i\eps)=\Theta(1)$, $\frac{1}{1+O(\eps)}\leq \frac{Q_i(v)}{Q_0(v)},\frac{F_i(v)}{F_0(v)}\leq 1+O(\eps)$. By Lemma~\ref{lem:query-complexity}, to distinguish $F_i$ and $F_0$,  $\Omega(\frac{1}{\eps^2})$ pricing queries on $\frac{1}{2}+4i\eps$ are needed.

Consider a setting where the underlying value distribution is $F_i$ for $i$ uniformly selected from $[\frac{1}{3}m,\frac{2}{3}m]$ but unknown beforehand. To find out the optimal revenue and the corresponding price, it is equivalent to find out the underlying value distribution. To distinguish $F_0$ and $F_i$, at least $\Omega(\frac{1}{\eps^2})$ pricing queries are needed on price $\frac{1}{2}+4i\eps$. Thus to distinguish $F_0$ and all other distributions, at least $\Omega(\frac{1}{\eps^2})$ pricing queries are needed on every price $\frac{1}{2}+4i\eps$, thus the query complexity is $\Omega(\frac{1}{\eps^3})$.
\end{proof}

\end{document}